\documentclass[journal]{IEEEtran}
\usepackage{amssymb}
\usepackage{mathrsfs}
\usepackage{graphicx}
\usepackage{amsmath, amssymb, amsthm}
\usepackage{leftidx}
\usepackage{extarrows}
\usepackage{stfloats}
\usepackage{picinpar}
\usepackage{enumerate}
\usepackage{algorithm}
\usepackage{algorithmicx}
\usepackage{algpseudocode}
\usepackage{epsfig}
\usepackage{latexsym}
\usepackage{amsfonts}
\usepackage{enumerate}
\usepackage{graphics}
\usepackage{graphicx,subfigure}
\usepackage{MnSymbol}
\usepackage{float}
\usepackage{pict2e}
\usepackage{tikz}
\usepackage{bm}
\newtheorem{theorem}{Theorem}
\newtheorem{assumption}{Assumption}

\newtheorem{lemma}{Lemma}
\newtheorem{remark}{Remark}
\newtheorem{definition}{Definition}

\title{Sensors Design for Large-Scale Boolean Networks via Pinning Observability}

\author{Shiyong~Zhu, Jianquan Lu$^\ast$, Jie Zhong, Yang Liu, and Jinde Cao 
\thanks{\textcolor[rgb]{0,0,1}{The revised version of this work has been accepted by IEEE Transactions on Automatic Control, DOI: 10.1109/TAC.2021.3110165.}}
\thanks{Shiyong Zhu is with the College of Mathematics and Computer Science, Zhejiang Normal University, Jinhua 321004, China, and also with the Department of Systems Science, School of Mathematics, Southeast University, Nanjing 210096, China (email: zhusy0904@gmail.com).}
\thanks{Jianquan Lu is with the Department of Systems Science, School of Mathematics, Southeast University, Nanjing 210096, China (email: jqluma@seu.edu.cn).}
\thanks{Jie Zhong and Yang Liu are with the College of Mathematics and Computer Science, Zhejiang Normal University, Jinhua 321004, China (e-mail: zhongjie0615@gmail.com;liuyang@zjnu.edu.cn).}
\thanks{Jinde Cao is with the School of Mathematics, Frontiers Science Center for Mobile Information Communication and Security, Southeast University, Nanjing 210096, China, with the Purple Mountain Laboratories, Nanjing 211111, China, and also with the Yonsei Frontier Lab, Yonsei University, Seoul 03722, South Korea. (e-mail: jdcao@seu.edu.cn).}
}

\begin{document}
\maketitle
\thispagestyle{empty}
\pagestyle{empty}
\begin{abstract}
  In this paper, a set of sensors is constructed via the pinning observability approach with the help of observability criteria given in \cite{Margaliot2019TAC2727} and \cite{Margaliot2019IEEECSL210}, in order to make the given Boolean network (BN) be observable. Given the assumption that system states can be accessible, an efficient pinning control scheme is developed to generate an observable BN by adjusting the network structure rather than just to check system observability. Accordingly, the sensors are constructed, of which the form is consistent with that of state feedback controllers in the designed pinning control. Since this pinning control approach only utilizes node-to-node message communication instead of global state space information, the time complexity is dramatically reduced from $O(2^{2n})$ to $O(n^2+n2^d)$, where where $n$ and $d$ are respectively the node number of the considered BN and the largest in-degree of vertices in its network structure. Finally, we design the sensors for the reduced D. melanogaster segmentation polarity gene network and the T-cell receptor kinetics, respectively.
\end{abstract}


\section{Introduction}\label{section-introduction}
Boolean networks (BNs) can be regarded as a classical type of discrete-time dynamical systems with binary nodal states, where systems state updates in the light of several pre-assigned logical functions \cite{kauffman1969jtb437}. Even if the expression of a BN is simple, its applications have been congruously recognized in numerous fields, such as biology engineering, power systems, circuit systems, and so forth.

In 2009, Cheng {\em et al.} proposed the algebraic state space representation (ASSR) of BNs based on semi-tensor product (STP) of matrices (cf. \cite{chengdz2011springer}), then it uplifted the scientific stream on this area to a new height. With the help of ASSR approach, various interesting works on BNs have taken place up to now, including but not limited to, controllability and observability \cite{margaliot2012aut1218,lujq2016ieeetac1658,mengmin2019auto70,zhangkz2020tac,zhangzh2020tac,zhusy2021tac,liht2020siam3632}, stability and stabilization \cite{valcher2015aut21,lir2013ieeetac1853,huangc2021tac,zhaogd2020iet2566,zhusy2019tac}, state estimation \cite{chenhw2020tac}, optimal control \cite{wuyh2018tac262,valcher2013tac1258}, detectability \cite{wangb2020scl104783} and synchronization \cite{linlin2020JFI}. Unfortunately, ASSR method runs in exponential time with respect to the node number $n$ of the considered BN, because this approach is established on the $2^n\times2^n$-dimensional network transition matrix $L$. Such high computational complexity limits almost works built on the ASSR method to BNs with a mass of nodes, though many problems of BNs have been proven to be NP-hard and the obtained results are theoretically prefect. However, large-scale BNs are universal in practice, thus many researchers begun to focus on reducing the time complexity of the existing approaches (see, e.g., \cite{liht2021tac,zhangkz2020tac,zhongjie2019new}).

It is of both theoretical interest and practical significance to study system observability. Since sensors cannot be imposed on all nodes specially for BNs with a mass of nodes, while the system state is not directly known, we are interested in using the output observations to fully infer the initial system state; this further permits us to reconstruct the entire state trajectory that is compatible with the output observations. As formally stated, a BN with such ability is said to be observable. Additionally, an observable system is also the precondition of designing state feedback controllers \cite{lir2013ieeetac1853} and its eliciting problems like optimal control \cite{wuyh2018tac262,valcher2013tac1258}. In the remarkable paper \cite{Margaliot2013Observability2351}, Laschov {\em et al.} proved that checking the observability of BNs is an $\text{NP}$-hard problem, and proposed the T-initial condition observable graph to characterize the observability of BNs on the basis of the observable graphs established in \cite{jungers2011observable}. To date, many different approaches have emerged to study the observability of BNs and Boolean control networks, such as matrix testing approach \cite{Valcher2012TAC1390}, automata language approach \cite{zhangkz2015IEEETAC2733}, and set reachability/controllability method \cite{chengdz2018scl22,guoyq2018tnnls6402}. However, the time complexity of the existing methods remains $O(2^{2n})$ at least due to the utilization of ASSR approach. As a result, it was also stressed in \cite{zhangkz2020tac} that many existing approaches built on the ASSR framework generally cannot address the observability of a BN with more than approximately $30$ nodes in a reasonable amount of time in practice. Thus, the efficient approach to analyze the observability of large-scale BNs is urgently necessary, and this constitutes the primary purpose of this article. Most recently, some complexity reduction techniques on the observability of BNs (including, e.g. \cite{Margaliot2019TAC2727}, \cite{Margaliot2019IEEECSL210}) have been established. In \cite{Margaliot2019TAC2727}, based on the one-to-one dependence between any conjunctive BN and network structure, Weiss and Margaliot proposed a polynomial-time algorithm to test the observability of conjunctive BNs and, for which, optimality with respect to the number of sensors needed for the observability was further attained. While following the similar spirit, they also found that the network-structure-based criteria in \cite{Margaliot2019TAC2727} suffice to ensure the observability of general BNs \cite{Margaliot2019IEEECSL210}, but only the suboptimality can be achieved. Although the developed conditions in \cite{Margaliot2019IEEECSL210} are only sufficient, it potentially motivates us to design efficient observers via the pinning observability viewpoint that will be elaborated in this article, because distributed pinning controllers have been proposed to modify the network structures of BNs in \cite{zhongjie2019new}.

When a BN is not observable, a natural problem is how to add sensors to make the considered system be observable. Even more, it was stressed in \cite{mathsystemtheory} that system reconstructibility is necessary for the existence of fully-ordered state observers in a linear time-invariant system and this claim is also true for BNs. Note that many efforts have been made to design distinct types of observers for BNs under the assumption that the considered BNs/BCNs are reconstructible \cite{Valcher2012TAC1390,zhangzh2020tac,yangjq2020tcns}, whereas it is evident that not all systems are reconstructible, in which case, adding sensors to make a BN be reconstructible or observable is meaningful. Besides, the methods in \cite{Valcher2012TAC1390,zhangzh2020tac,yangjq2020tcns} are inevitable to encounter the $2^n \times 2^n$-dimensional state space, which eventually gives rise to the exponential time complexity. Noticing that the results in \cite{Margaliot2019TAC2727} and \cite{Margaliot2019IEEECSL210} are polynomial-time, along this line, we attempt to propose a new sensors design method from the aspect of pinning observability, which as indeed is an alternative of sensors design and is an application of observability criteria proposed in \cite{Margaliot2019TAC2727} and \cite{Margaliot2019IEEECSL210}. In the area of BNs, the first attempt to design pinning controllers was made in \cite{Liff2016TNNLS1585} for sake of stabilizing a BN at an appointed equilibrium point. However, except for the exponential time complexity, traditional pinning approaches \cite{Liff2016TNNLS1585,liff2018tnnls} is also difficult to identify pinned nodes. Therefore, here we utilize the pinning controllers designed in \cite{zhongjie2019new} in which a part of controllers is also built on the network structures.

In this article, in order to design the sensors making the BNs observable, we start with supposing that the state of arbitrary node is accessible. Then, a network-structure-based pinning control strategy is proposed. Based on Properties $P_1$ and $P_2$ in \cite{Margaliot2019TAC2727} and \cite{Margaliot2019IEEECSL210}, a polynomial-time algorithm is established to identify the pinned nodes and to print $p$ desired observed paths, where $p$ is the number of outputs in the original system. Subsequently, the state feedback gain imposed on each pinned node is designed to modify the adjacency relation of network structure. During the whole procedure, the dimensions of addressed matrices are bounded by $2\times 2^{d+1}$ instead of $2^n \times 2^n$, where $d$ is the largest in-degree of vertices in network structure. Finally, we construct the feasible sensors based on the designed pinning controller; together with original outputs, the initial state of the considered BN can be inferred in accordance with the output observations. In summary, the contributions of this paper are concluded from three aspects:
\begin{itemize}
  \item[(1)] We do not utilize the $2^n\times 2^n$-dimensional network transition matrix $L$ in this paper. Compared with the approaches in \cite{Valcher2012TAC1390,zhangzh2020tac,yangjq2020tcns}, by utilizing the observability criteria obtained in \cite{Margaliot2019TAC2727} and \cite{Margaliot2019IEEECSL210}, the time complexity of sensors design is reduced from $O(2^{2n})$ to $O(n^2+n2^d)$. Specially for the sparsely connected large-scale biological networks (i.e., $n\gg d$), it follows that $O(n^2+n2^{d}) \ll O(2^{2n})$.
  \item[(2)] In comparison with the traditional pinning control approach (see, e.g., \cite{Liff2016TNNLS1585,liff2018tnnls}), our approach can determine the pinned nodes in time $O(n^2)$. Moreover, it does overcome the drawback of the traditional approaches where pinned node set may contain all network nodes.
  \item[(3)] Different from the observers design in \cite{Valcher2012TAC1390,zhangzh2020tac,yangjq2020tcns}, we construct the sensors from the viewpoint of pinning observability to make the considered system be observable.
\end{itemize}

The remainder of this article is organized as follows. Section \ref{section-preliminary} presents some preliminaries, and Section \ref{section-problemformulation} describes the considered problem. In Section \ref{section-openov}, a network-structure-based pinning control strategy is developed to enforce a unobservable BN to be observable. Moreover, Section \ref{section-sensors} designs the sensors to estimate the initial state of the considered BN. Two biological examples are addressed in Section \ref{section-example}. Finally, Section \ref{section-conclusion} concludes this article.

\section{Preliminaries}\label{section-preliminary}
Throughout this article, the following notations will be utilized to facilitate the readers. Let $\mathbb{R}$ and $\mathbb{N}$ represent the sets of real numbers and positive integers, respectively. Specially, $\mathscr{B}:=\{1,0\}$. $\mathbb{R}^{m\times n}$ stands for the set of all $m\times n$-dimensional real matrices. Given any integers $i,j\in\mathbb{N}$ with $i<j$, we denote by $[i,j]$ the set $\{i,i+1,\cdots,j\}$. The $n\times n$-dimensional identity matrix is denoted by $I_n$, and its $i$-th column is denoted as $\delta_n^i$. Letting $\Delta_n:=\{\delta_n^i\mid i\in[1,n]\}$, a matrix $A \in \mathbb{R}^{m\times n}$ is called a logical matrix if every column therein, denoted by $\text{Col}_i(A)$, lies in set $\Delta_m$. All $m\times n$-dimensional logical matrices are collected by $\mathscr{L}_{m\times n}$, and every logical matrix $[\delta_{m}^{i_1},\delta_{m}^{i_2},\cdots,\delta_m^{i_n}]$ therein can be briefly written as $\delta_{m}[i_1,i_2,\cdots,i_n]$. Given a matrix $A$, matrix $A^\top$ captures its transpose. Given a set $B$, $\mid B \mid$ describes its cardinality. ${\bf 1}_{n}$ is the $n$-dimensional column vector with all entries being $1$, i.e., ${\bf 1}_{n}:=\sum_{i=1}^{n}\delta_n^i$. Logical operators $\neg$, $\vee$, $\wedge$, $\bar{\vee}$ and $\leftrightarrow$ are respectively negation, disjunction, conjunction, exclusion, and bicondition.

\begin{definition}[See \cite{chengdz2011springer}]
Given matrices $U\in\mathbb{R}^{m\times n}$ and $V\in\mathbb{R}^{p\times q}$, the STP of matrices $U$ and $V$ is defined as
$U\ltimes V:=(U\otimes I_{\phi(n,p)/n})(V\otimes I_{\phi(n,p)/p})$,
where ``$\otimes$'' is the tensor product, and $\phi(n,p)$ is denoted as the least common multiple of integers $n$ and $p$. In particular, if matrices $U$ and $V$ satisfy $n=p$, their STP becomes conventional matrix product.
\end{definition}

Then, several properties of STP of matrices are briefly introduced. We refer the readers to \cite{chengdz2011springer} and references therein for a survey of STP of matrices and the multilinear form of logical functions.
\begin{lemma}[See \cite{chengdz2011springer}]\label{lemma-swapmatrix}
The STP of matrices has the following useful properties:
\begin{itemize}
  \item[(1)] Given $u\in\mathbb{R}^{m\times 1}$ and $v\in \mathbb{R}^{n\times 1}$, then $u \ltimes v= \mathrm{W}_{[n,m]}\ltimes {\bm v}\ltimes {\bm u}$ holds, where $\mathrm{W}_{[n,m]}$ is a swap matrix defined as
      $\mathrm{W}_{[n,m]}:=[I_m\otimes\delta_n^1,I_m\otimes\delta_n^2,\cdots,I_m\otimes \delta_n^n]$.
  \item[(2)] Given $u\in\mathbb{R}^{m\times 1}$ and $A\in\mathbb{R}^{p\times q}$, then $u \ltimes A= (I_{m}\otimes A)\ltimes u$ holds;

  \item[(3)] Given $M_{r,n}:=[\delta_{n}^1\otimes\delta_{n}^1~\delta_{n}^2\otimes\delta_{n}^2~\cdots~\delta_{n}^{n}\otimes\delta_{n}^{n}]$, then $\eta\ltimes\eta=M_{r,n}\eta$ holds for every $\eta\in\Delta_{n}$;

  \item[(4)] Given $M_{d,n}:={\bf 1}^\top_n \otimes I_n$, then $v=M_{d,n}\ltimes u \ltimes v$ holds.
\end{itemize}
\end{lemma}

Defining a bijection $\varrho:\mathscr{B}\rightarrow\Delta_2$ as $x_i:=\varrho({\bm x}_i):=\delta_2^{2-{\bm x}_i}$, the one-to-one correspondence can be given as $\zeta({\bm x})=x:=\ltimes_{i=1}^n x_i$, for all ${\bm x}=({\bm x}_1,{\bm x}_2,\cdots,{\bm x}_n)\in\mathscr{B}^n$. Accordingly, arbitrary logical function ${\bm f}({\bm x}_1,{\bm x}_2,\cdots,{\bm x}_n):\mathscr{B}^n\rightarrow\mathscr{B}$ can be equivalently expressed in multi-linear form $f(x_1,x_2,\cdots,x_n):(\Delta_2)^n\rightarrow\Delta_2$.
\begin{lemma}[See \cite{chengdz2011springer}]\label{lemma-structurematrix}
Given function ${\bm f}({\bm x}_1,{\bm x}_2,\cdots,{\bm x}_n):\mathscr{B}^n\rightarrow \mathscr{B}$, there exists a unique logical matrix $L_{\bm f}\in\mathscr{L}_{2\times 2^n}$ such that
\begin{equation}\label{equation-multilinear}
f(x_1,x_2,\cdots,x_n)=L_{\bm f} x,
\end{equation}
where matrix $L_{\bm f}\in\mathscr{L}_{2\times 2^n}$ is called the structure matrix of logical function ${\bm f}$, and equation (\ref{equation-multilinear}) is called the multilinear form of ${\bm f}$.
\end{lemma}

Next, the functional variables of a logical function are defined.
\begin{definition}\label{definition-functional}
Given a logical function ${\bm f}({\bm x}_1,{\bm x}_2,\cdots,{\bm x}_n):\mathscr{B}^n\rightarrow\mathscr{B}$, variable ${\bm x}_\kappa$ is called functional if ${\bm f}({\bm x})\neq {\bm f}({\bm x}+(\delta_n^\kappa)^\top)$ holds for some ${\bm x}\in\mathscr{B}^n$ with ${\bm x}_\kappa=0$. Otherwise, variable ${\bm x}_\kappa$ is said to be nonfunctional.
\end{definition}

Based on the multilinear form $f$, one can present the equivalent algebraic description of Definition \ref{definition-functional}.
\begin{lemma}
Given a logical function ${\bm f}({\bm x}_1,{\bm x}_2,\cdots,{\bm x}_n): \mathscr{B}^n\rightarrow \mathscr{B}$, state variable ${\bm x}_\kappa$ is nonfunctional if, for all $j\in[1,2^{n-1}]$, it holds that $\text{Col}_j(L_{\bm f} \mathrm{W}_{[2,2^{\kappa-1}]})=\text{Col}_{j+2^{n-1}}(L_{\bm f} \mathrm{W}_{[2,2^{\kappa-1}]})$.
\end{lemma}
\begin{proof}
Using bijection $\zeta$ and STP of matrices, if state variable ${\bm x}_\kappa$ is nonfunctional, then
\begin{equation}\label{equation-equal}\begin{aligned}
L_{\bm f} \ltimes_{i=1}^{n} x_i&=L_{\bm f}\mathrm{W}_{[2,2^{\kappa-1}]} x_\kappa \ltimes (\ltimes_{i=1}^{\kappa-1}x_{i}) \ltimes (\ltimes_{i=\kappa+1}^{n}x_{i})\\
&=L_{\bm f} \mathrm{W}_{[2,2^{\kappa-1}]} ({\bf 1}_2-x_\kappa) \ltimes (\ltimes_{i=1}^{\kappa-1}x_{i}) \ltimes (\ltimes_{i=\kappa+1}^{n}x_{i})\\
\end{aligned}\end{equation}
holds for any $x_i\in\Delta_2$, $i\in[1,n]$. Letting $(\ltimes_{i=1}^{\kappa-1}x_{i}) \ltimes (\ltimes_{i=\kappa+1}^{n}x_{i}):=\delta_{2^{n-1}}^j$, it is claimed from equation (\ref{equation-equal}) that
$$L_{\bm f} \mathrm{W}_{[2,2^{\kappa-1}]} \delta_2^1\delta_{2^{n-1}}^j = L_{\bm f}\mathrm{W}_{[2,2^{\kappa-1}]} \delta_2^2 \delta_{2^{n-1}}^j.$$

Since $\delta_2^1\ltimes\delta_{2^{n-1}}^j=\delta_{2^n}^j$ and $\delta_2^2\ltimes\delta_{2^{n-1}}^j=\delta_{2^n}^{2^{n-1}+j}$ hold for all $j\in[1,2^{n-1}]$, one can conclude that
$$\text{Col}_j(L_{\bm f} \mathrm{W}_{[2,2^{\kappa-1}]})=\text{Col}_{j+2^{n-1}}(L_{\bm f} \mathrm{W}_{[2,2^{\kappa-1}]}).$$
\end{proof}

Finally, we give some basic concepts of graph theory. A digraph is denoted by an ordered pair $\mathrm{G}:=(\mathrm{V},\mathrm{E})$, where $\mathrm{V}:=\{v_1,v_2,\cdots,v_n\}$ is the vertex set and $\mathrm{E}$ means the set of directed edges. More precisely, $e_{ij}=(v_i,v_j)\in \mathrm{E}$ is an oriented edge drawn from vertex $v_i$ to vertex $v_j$. For every $e_{ij}:=(v_i,v_j)\in \mathrm{E}$, edge $e_{ij}$ is said to joint vertices $v_i$ to $v_j$, that is, $v_i$ dominates $v_j$. For vertex $v_i$, all vertices that dominate (resp., are dominated by) it are said to be its in-neighbors (resp., out-neighbors). which are collected by set $I(i)$ (resp., $O(i)$). An ordered vertex sequence $v_{i_0}v_{i_1}\cdots v_{i_s}$ is a walk in a digraph $\mathrm{G}$ from $v_{i_0}$ to $v_{i_s}$ if $e_{i_t i_{t+1}}\in \mathrm{E}$ for all $t\in[0,s-1]$. Moreover, it becomes a path if vertices therein are pairwise different. A path becomes a cycle if $i_0=i_s$.

\section{Problem Formulation and Research Framework}\label{section-problemformulation}
\subsection{Problem Formulation}
Consider the following BN with $n$ state nodes and $p$ output nodes:
\begin{equation}\label{equation-BN}
\left\{
\begin{aligned}
&{\bm x}_i(t+1)={\bm f}_i({\bm x}_{i_1}(t),{\bm x}_{i_2}(t),\cdots,{\bm x}_{i_{k_i}}(t)), ~i\in [1,n],\\
&{\bm y}_j(t)={\bm x}_j(t), ~j\in[1,p],
\end{aligned}
\right.
\end{equation}
where ${\bm x}_i(t)\in \mathscr{B}$ and ${\bm y}_j(t)\in\mathscr{B}$ stand for system state variables and output variables, respectively. Without loss of generality, the first $p$ state variables are assumed to be directly observable, or alternatively variables ${\bm x}_1,{\bm x}_2,\cdots,{\bm x}_p$ are termed directly observable variables. In this article, ${\bm y}_j(t)={\bm x}_j(t)$, $j\in[1,p]$, are also called the original sensors of BN (\ref{equation-BN}). It should be stressed that functions ${\bm f}_i$ and ${\bm h}_j$ are all minimally represented. For BN (\ref{equation-BN}), our purpose is to design sensors in the following form such that, together with the designed sensors, BN (\ref{equation-BN}) will be observable:
\begin{equation}\label{equation-observer}
\tilde{{\bm y}}_j(t)={\bm h}_j({\bm x}_{1}(t),{\bm x}_{2}(t),\cdots,{\bm x}_{n}(t)), ~j\in[1,\tilde{p}],
\end{equation}
where integer $\tilde{p}$ and logical functions ${\bm h}_j$, $j\in[1,\tilde{p}]$, are the configurations that need to be designed.

In particular, if BN (\ref{equation-BN}) has been observable under the given sensors, the original sensors have been enough to ensure the system observability. In this case, the result is trivial, thus BN (\ref{equation-BN}) is assumed to be unobservable.
\begin{definition}\label{definition-observability}
BN (\ref{equation-BN}) is called observable if, for arbitrary distinct initial state pair ${\bm x}_0\in\mathscr{B}^n$ and ${\bm x}'_0\in\mathscr{B}^n$, there exists a time instant $\tau$ satisfying ${\bm y}({\bm x}_0,\tau) \neq {\bm y}({\bm x}'_0,\tau)$, where ${\bm y}({\bm x}_0,t)$ and ${\bm y}({\bm x}'_0,t)$ are the output sequences of BN (\ref{equation-BN}) with initial states ${\bm x}_0$ and ${\bm x}'_0$ respectively.
\end{definition}
\begin{assumption}\label{assumption-bn-is-unobservale}
BN (\ref{equation-BN}) is assumed to be unobservable.
\end{assumption}
\begin{remark}
In this article, we only address BN (\ref{equation-BN}) with directly observable outputs. For BNs with general output functions, one can follow the same extension technique as in [\cite{Margaliot2019TAC2727}, Sec. V.B] to convert them into the form of BN (\ref{equation-BN}).
\end{remark}

\subsection{Research Framework}
In \cite{Margaliot2019TAC2727} and \cite{Margaliot2019IEEECSL210}, several sufficient criteria were developed to judge the observability of BNs from the viewpoint of network structures. In the following, we briefly reviews some related concepts and lemmas.

Given BN (\ref{equation-BN}), by defining $\mathrm{V}=\{{v}_1,{v}_2,\cdots,{v}_{{n}}\}$ as the vertex set, its {\em network structure} can be described as an ordered pair $\mathrm{G}:=(\mathrm{V},\mathrm{E})$, where vertex $v_i$ corresponds to node $i$ of BN (\ref{equation-BN}), and $e_{ij} \in \mathrm{E}$ if and only if ${\bm x}_i$ is a functional variable of ${\bm f}_j$. Specially, since the states of vertices $v_i$, $i\in[1,p]$, are directly observed, we also term them as directly observable vertices.

\begin{definition}[See \cite{Margaliot2019TAC2727,Margaliot2019IEEECSL210}]\label{definition-observedpath}
Given network structure $\mathrm{G}$, a path $P:{v}_{i_0}{v}_{i_1}\cdots {v}_{i_s}$ is said to be an observed path if, vertex ${v}_{i_s}$ is the unique directly observable vertex on path $P$ and vertex ${v}_{i_k}$ is the only in-neighbor of ${v}_{i_{k+1}}$ for each $k\in[0,s-1]$.
\end{definition}

\begin{definition}[See \cite{Margaliot2019TAC2727,Margaliot2019IEEECSL210}]
BN (\ref{equation-BN}) is said to have Property $P_1$ if, for every non directly observable vertex ${v}_i$ in its network structure $\mathrm{G}$, there exists a vertex ${v}_j \neq v_i$ with $I(j)=\{v_i\}$. BN (\ref{equation-BN}) is said to have Property $P_2$ if, for each cycle $C$ entirely composed of non directly observable vertices, there exists a vertex ${v}_i\in C$ such that ${v}_i$ is the unique in-neighbor of a vertex ${v}_j\not\in C$.
\end{definition}

\begin{lemma}[See \cite{Margaliot2019TAC2727,Margaliot2019IEEECSL210}]\label{lemma-basictheorem}
BN (\ref{equation-BN}) is observable if it has Properties $P_1$ and $P_2$; or alternatively, its network structure $\mathrm{G}$ can be decomposed into several disjoint observed paths.
\end{lemma}

Lemma \ref{lemma-basictheorem} provides an observability criterion from the viewpoint of network structure $\mathrm{G}$, which only contains $n$ vertices. In the following, we shall sketch that this sufficient condition is useful for the sensors construction in a high-level description, and the detailed proof will be elaborated in Section \ref{section-sensors}. Thus, the results in this article can be regarded as an application of those in \cite{Margaliot2019TAC2727} and \cite{Margaliot2019IEEECSL210}. Observing Definition \ref{definition-observedpath} and Lemma \ref{lemma-basictheorem}, if a vertex $v_i$ has an out-neighbor $v_j$ satisfying $I(j)=\{i\}$ in network structure $\mathrm{G}$, then state variable ${\bm x}_i(\kappa)$ at time instant $\kappa$ will be passed to variable ${\bm x}_j(\kappa+1)$ at the next time instant, and ${\bm x}_j(\kappa+1)$ can also inversely determine variable ${\bm x}_i(\kappa)$, because both of logical mappings ${\bm x}_j(t+1)={\bm x}_i(t)$ and ${\bm x}_j(t+1)=\overline{{\bm x}_i(t)}:=\neg{\bm x}_i(t)$ are bijective. Thus, while there exists an observed path stemming from vertex $v_i$ and terminating at a directly observable vertex $v_{\sigma_i}$ as in the left subgraph of Fig. \ref{fig-ideagraph}, the value of variable ${\bm x}_i$ at any time instant is inferable even if it is not directly observable. However, if each out-neighbor $v_j$ of vertex $v_i$ does not satisfy $I(j)=\{i\}$ but vertices on the path from $v_j$ to $v_{\sigma_i}$ are eligible as in the right subgraph of Fig. \ref{fig-ideagraph}, one can design sensors ${\bm u}_j={\bm g}_j({\bm x}_{j_1},{\bm x}_{j_2},\cdots,{\bm x}_{j_{k_j}})$; together with the state of vertex $v_j$ inferred by ${\bm y}_{\sigma_i}$, these sensors can determine the value of ${\bm x}_i$. In other words, utilizing some logical speculation $\oplus_j$ as ${\bm u}_j \oplus_j {\bm f}_j$ can imply the value of ${\bm x}_i$.
\begin{figure}[!ht]
\centering
\includegraphics[width=0.2\textwidth=0.4]{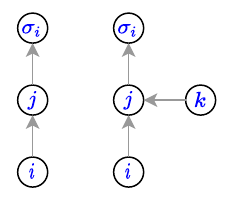}
\caption{Two different cases of state passing from node $i$ to directly observable node $\sigma_i$ are presented, where the left subgraph describes the situation where vertex $v_i$ has an out-neighbor $v_j$ satisfying $I(j)=\{i\}$ in network structure $\mathrm{G}$ while the right one captures the remainder scene. }\label{fig-ideagraph}
\end{figure}

By regarding the vertex corresponding to pinned node $j$, logical function ${\bm g}_j$ and logical speculation $\oplus_j$ as a pinned node, the feedback control gain and the logical coupling, it is given in the same form as pinning control, even if we know that not all state variables are accessible when a BN is not observable. Fortunately, for the sensors construction, we are admissible to obtain some extra state observations. Keep it in mind, the following research starts with establishing a novel pinning control strategy to make an arbitrary unobservable BN (4) be observable by the assumption that state variables are fully accessible, and proceed to use the obtained controllers to construct efficient sensors (5).

\section{An Improved Pinning Strategy For Observability}\label{section-openov}
In this section, a network-structure-based pinning control strategy will be developed so as to make a unobservable BN (\ref{equation-BN}) become observable. Although it seems illogical as system state is not available at any time instant, as stressed above, such pinning control strategy for observability paves the way for the sensors construction of large-scale BNs.

The pinning controlled BN (\ref{equation-BN}) can be presented as follows:
\begin{equation}\label{equation-pinningBN}
\left\{
\begin{aligned}
&{\bm x}_i(t+1)={\bm f}_i({\bm x}_{i_1}(t),{\bm x}_{i_2}(t),\cdots,{\bm x}_{i_{k_i}}(t)), ~i\in [1,n]\backslash\mathrm{P},\\
&{\bm x}_i(t+1)=\tilde{{\bm f}}_i ({\bm u}_i(t),{\bm x}_{i_1}(t),{\bm x}_{i_2}(t),\cdots,{\bm x}_{i_{k_i}}(t)), ~i\in \mathrm{P},\\
&{\bm y}_j(t)={\bm x}_j(t), ~j\in[1,p],
\end{aligned}
\right.
\end{equation}
where $\mathrm{P}\subseteq [1,n]$ is the pinned node set. For every $i\in\mathrm{P}$, $\tilde{{\bm f}}_i:\mathscr{B}^{k_i+1}\rightarrow \mathscr{B}$ is a logical function, and input ${\bm u}_i:\mathscr{B}^{n} \rightarrow \mathscr{B}$ is state feedback control written as
\begin{equation}
{\bm u}_i(t)={\bm g}_i({\bm x}_{1}(t),{\bm x}_{2}(t),\cdots,{\bm x}_{n}(t)), ~i\in\mathrm{P}.
\end{equation}
Moreover, function $\tilde{{\bm f}}_i$, $i\in\mathrm{P}$ is expressed as in \cite{Liff2016TNNLS1585}: $\tilde{{\bm f}}_i={\bm u}_i \oplus_i {\bm f}_i$.

Three configurations that we shall design in (\ref{equation-pinningBN}) are pinned node set $\mathrm{P}$, state feedback law ${\bm g}_i$, as well as binary logical operator $\oplus_i$.

Although Lemma \ref{lemma-basictheorem} is only sufficient, it suffices to design pinning controller, by which one can modify the adjacency relation of BN (\ref{equation-BN}) such that the resulting BN (\ref{equation-pinningBN}) satisfies Properties $P_1$ and $P_2$.

In what follows, for a unobservable BN (\ref{equation-BN}), we design the pinning controller to achieve its observability by the following two steps:
\begin{itemize}
  \item[(1)] Select the pinned node set $\mathrm{P}$. We print $p$ desired observed paths $\mathrm{O}_1, \mathrm{O}_2,\cdots,\mathrm{O}_p$, then determine the pinned node set $\mathrm{P}$ to collect the nodes corresponding to the vertices that do not satisfy the definition of observed paths.
  \item[(2)] Design state feedback control and logical couplings. For every pinned node $i\in\mathrm{P}$, we aim to design the logical operator $\oplus_i$ and the state feedback gain ${\bm g}_i$ such that BN (\ref{equation-pinningBN}) satisfies both Properties $P_1$ and $P_2$.
\end{itemize}

\subsection{Identifying pinned node set $\mathrm{P}$}
In this subsection, we devote to selecting the pinned node set $\mathrm{P}$. According to Lemma \ref{lemma-basictheorem}, we aim to decompose the network structure of pinning controlled BN (\ref{equation-pinningBN}) into $p$ disjoint observed paths. It indicates if some ``control effect" is inflicted on some nodes such that the resulting network structure can be decomposed into several disjoint observed paths, one can generate an observable BN (\ref{equation-pinningBN}).

We establish Algorithm \ref{algorithm-selectnode}, which can be implemented in time $O(n^2)$, to print $p$ desired observed paths and identify the pinned node set $\mathrm{P}$. The printed paths $\mathrm{O}_1,\mathrm{O}_2,\cdots,\mathrm{O}_p$ may not satisfy Definition \ref{definition-observedpath}, thus they are termed the desired observed paths here.

\begin{algorithm}[ht!]
\caption{Pinned Node Set Search}\label{algorithm-selectnode}
\begin{algorithmic}[1]
\Require Network structure $\mathrm{G}:=(\mathrm{V},\mathrm{E})$ of BN (\ref{equation-BN})
\Ensure Pinned node set $\mathrm{P}$ and $p$ desired paths $\mathrm{O}_1, \mathrm{O}_2,\cdots, \mathrm{O}_p$
\Procedure{SEARCH}{$\mathrm{P},\mathrm{O}_1, \mathrm{O}_2,\cdots, \mathrm{O}_p$}
\State Set $\mathrm{P}:=\emptyset$ (*$\mathrm{P}$ is a set*)
\State Color all vertices in $\mathrm{V}$ white
\For{$i=1$ to $p$}
\State Set $\mathrm{O}_i:=\emptyset$ (*$\mathrm{O}_i$ is a list*)
\State Append $v_i$ to the tail of $\mathrm{O}_i$
\State Color vertex $v_i$ black
\EndFor
\For{$j=1$ to $p$}
\State Consider the tail $x$ of $\mathrm{O}_j$
\While{the tail $x$ of $\mathrm{O}_j$ has a white in-neighbor $y$}
\State $x=$\Call{Append}{$j,x,y$}
\EndWhile
\EndFor
\While{$\mathrm{V}$ has a white vertex $w$}
\State Consider the tail $z$ of $\mathrm{O}_p$
\State $z=$\Call{Append}{$p,z,w$}
\EndWhile
\If{$\mathrm{P}=\emptyset$}
\State\Return ``BN (\ref{equation-BN}) is observable''
\Else
\State\Return $(\mathrm{P},\mathrm{O}_1, \mathrm{O}_2,\cdots, \mathrm{O}_p)$
\EndIf
\EndProcedure

\Function {Append}{$j,x,y$}
\If{$I(x)\neq\{y\}$}
\State Set $\mathrm{P}:=\mathrm{P} \cup \{j|v_j=x\}$
\EndIf
\State Append $y$ to the tail of $\mathrm{O}_j$
\State Color $y$ black
\State\Return $y$
\EndFunction
\end{algorithmic}
\end{algorithm}

The employing principle of Algorithm \ref{algorithm-selectnode} is to establish $p$ disjoint desired observed paths. The basic idea is to check whether the visiting vertex has a unique in-neighbor according to the Line $26$ of Algorithm \ref{algorithm-selectnode}. Once the desired observed paths have been printed, the nodes that correspond to the vertices violating Definition \ref{definition-observedpath} are recognized as the pinned nodes.

\subsection{Designing $\oplus_i$ and ${\bm g}_i$ for every $i\in\mathrm{P}$}
To proceed, for different types of pinned nodes in $\mathrm{P}$, we design the corresponding logical operators $\oplus_i:\mathscr{B}^2\rightarrow\mathscr{B}$ and Boolean functions ${\bm g}_i$ to modify the network structure of BN (\ref{equation-BN}) such that the conditions in Lemma \ref{lemma-basictheorem} are satisfied. Accordingly, we split the set $\mathrm{P}$ into three categories $\mathrm{P}:=\mathrm{P}_1\cup\mathrm{P}_2\cup\mathrm{P}_3$ as indicated in Fig. \ref{fig-class}, where for node $i\in\mathrm{P}_1$ we should both delete and add incoming edges; for node $i\in\mathrm{P}_2$ we need to delete some edges ending with it; and for $i\in\mathrm{P}_3$ we need to add an edge from another node to it.
\begin{figure}[ht!]
\centering
\includegraphics[width=0.3\textwidth=0.5]{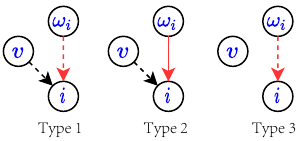}
\caption{Categorization on the pinned nodes is presented, where vertex $v_{\omega_i}$ is the desired in-neighbor of vertex $v_i$ on the desired observed paths and vertex $v$ simply stands for another in-neighbor of vertex $v_i$. More precisely, black dashed edges correspond to the edges needing to be deleted, and red dashed (resp., red solid) edges represent the edges that need to be added (resp., unchanged).}\label{fig-class}
\end{figure}

\begin{lemma}\label{propostion-solution}
Given logical function ${\bm f}:\mathscr{B}^n\rightarrow\mathscr{B}$, whose structure matrix is denoted as $L_{{\bm f}}\in\mathscr{L}_{2\times 2^{n}}$. Variable ${\bm x}_{\kappa}$ is a unique functional variable of function ${\bm f}$ if and only if structure matrix $L_{{\bm f}}\in\mathscr{L}_{2\times 2^n}$ satisfies that
\begin{equation}\label{equation-beforeset}
L_{{\bm f}}\ltimes \mathrm{W}_{[2,2^{\kappa-1}]}=\hat{A}\otimes {\bf 1}^\top_{2^{n-1}},
\end{equation}
where $\hat{A}=I_2$ or $\hat{A}=I'_2:=\delta_2[2,1]$.
\end{lemma}
\begin{proof}
Consider the multi-linear form $f$ of function ${\bm f}$. By using the STP of matrices, one has that
\begin{equation*}
\begin{aligned}
{f}(x_{1},x_{2},\cdots,x_{n})&=L_{\bm f}\ltimes_{j=1}^{n} x_j=L_{{\bm f}}\ltimes_{j=1}^{\kappa-1}x_j \ltimes x_{\kappa} \ltimes_{h=\kappa+1}^{n}x_j\\
&=L_{{\bm f}} \mathrm{W}_{[2,2^{\kappa-1}]}x_{\kappa}(\ltimes_{j=1}^{\kappa-1}x_j)(\ltimes_{j=\kappa+1}^{n}x_j).
\end{aligned}
\end{equation*}

As variables $x_j$, $j\neq\kappa$ are nonfunctional, then it holds that
$$ L_{{\bm f}} \mathrm{W}_{[2,2^{\kappa-1}]}\delta_2^1 \delta_{2^{n-1}}^j = L_{{\bm f}} \mathrm{W}_{[2,2^{\kappa-1}]}\delta_2^1 \delta_{2^{n-1}}^k$$
and
$$ L_{{\bm f}} \mathrm{W}_{[2,2^{\kappa-1}]}\delta_2^2 \delta_{2^{n-1}}^j = L_{{\bm f}} \mathrm{W}_{[2,2^{\kappa-1}]}\delta_2^2 \delta_{2^{n-1}}^k$$
for any $j\neq \kappa$. Furthermore, since variable $x_\kappa$ is functional, one has that
$$ L_{{\bm f}} \mathrm{W}_{[2,2^{\kappa-1}]}\delta_2^1 = [1,0]^\top \otimes {\bf 1}^\top_{2^{n-1}},~ L_{{\bm f}} \mathrm{W}_{[2,2^{\kappa-1}]}\delta_2^2 = [0,1]^\top \otimes {\bf 1}^\top_{2^{n-1}},$$
or
$$ L_{{\bm f}} \mathrm{W}_{[2,2^{\kappa-1}]}\delta_2^1 = [0,1]^\top \otimes {\bf 1}^\top_{2^{n-1}},~ L_{{\bm f}} \mathrm{W}_{[2,2^{\kappa-1}]}\delta_2^2 = [1,0]^\top \otimes {\bf 1}^\top_{2^{n-1}}.$$

Therefore, it follows that $L_{{\bm f}}\mathrm{W}_{[2,2^{\kappa-1}]}=I_2\otimes {\bf 1}^\top_{2^{n-1}}$ or $L_{{\bm f}}\mathrm{W}_{[2,2^{\kappa-1}]}=I_2^\top\otimes {\bf 1}^\top_{2^{n-1}}$ hold for any logical function ${\bm f}$ with unique functional variable ${\bm x}_\kappa$, and vice versa.
\end{proof}

\begin{lemma}\label{lemma-solution}
For every $\hat{A}\in\mathscr{L}_{2\times 2}$ and $\kappa\in[1,n]$, there exists a unique matrix $A\in\mathscr{L}_{2\times 2^{n}}$ satisfying equation (\ref{equation-beforeset}), i.e., $A\ltimes \mathrm{W}_{[2,2^{\kappa-1}]}=\hat{A}\otimes {\bf 1}^\top_{2^{n-1}}$.
\end{lemma}
\begin{proof}
Following the definition of swap matrix $\mathrm{W}_{[2,2^{\kappa-1}]}$, one has that
$\mathrm{W}_{[2,2^{\kappa-1}]}\ltimes \mathrm{W}^\top_{[2,2^{\kappa-1}]}=I_{2^{\kappa}}$
and
$\mathrm{W}^\top_{[2,2^{\kappa-1}]}\in\mathscr{L}_{2^{\kappa}\times 2^{\kappa}}$. Therefore, if we multiply $\mathrm{W}^\top_{[2,2^{\kappa-1}]}$ on the both side of equation (\ref{equation-beforeset}), one can imply that
\begin{equation}\label{equation-ba}
A=(\hat{A}\otimes {\bf 1}^\top_{2^{n-1}}) \ltimes \mathrm{W}^\top_{[2,2^{\kappa-1}]}\in\mathscr{L}_{2 \times 2^n}
\end{equation}
holds for all $\hat{A}\in\mathscr{L}_{2\times 2}$. This completes the proof of Lemma \ref{lemma-solution}.
\end{proof}

Let $x_i(t):=\rho({\bm x}_i(t))$, $u_j(t):=\rho({\bm u}_j(t))$ and $y_k(t):=\rho({\bm y}_k(t))$. Utilizing the STP of matrices and its properties introduced in Section \ref{section-preliminary}, we can derive the multilinear form of nodal dynamics (\ref{equation-BN}) as
\begin{equation}
\left\{\begin{aligned}
&x_i(t+1)=L_{{\bm f}_i}\ltimes_{\alpha=1}^{k_i} x_{i_\alpha}(t),~i\in [1,n],\\
&y_j(t)=x_j(t),~j\in[1,p],
\end{aligned}\right.
\end{equation}
where $L_{{\bm f}_i}\in\mathscr{L}_{2\times 2^{k_i}}$ is the structure matrix of logical function ${\bm f}_i$.

For $i\in\mathrm{P}_1$, we denote by $v_{\varpi_i}$ the desired in-neighbor of vertex $v_i$ on the desired observed paths returned by Algorithm \ref{algorithm-selectnode}. In this case, state feedback control ${\bm u}_i$ is given as
$${\bm u}_i(t)={\bm g}_i({\bm x}_{i_1}(t),{\bm x}_{i_2}(t),\cdots,{\bm x}_{i_{k_i}}(t),{\bm x}_{\varpi_i}(t)).$$

Let $\varpi_i$ be the $\iota_i$-th number of set $\{i_1,i_2,\cdots,i_{k_i}\}$ ordered increasingly. Then, defining $\bar{\ltimes}_{\alpha=1}^{k_i} x_{i_\alpha}(t)=(\ltimes_{\alpha=1}^{\iota_i-1} x_{i_\alpha}(t)) \ltimes x_{\varpi_i}(t) \ltimes (\ltimes_{\alpha=\iota_i}^{i_\alpha} x_{i_\alpha}(t))$, one can obtain that
\begin{equation}\label{equation-p1}
\begin{aligned}
x_i(t+1)&=M_{\oplus_i}u_i(t){f}_i(x_{i_1}(t),x_{i_2}(t),\cdots,x_{i_{k_i}}(t))\\
&=M_{\oplus_i}L_{{\bm g}_i} \bar{\ltimes}_{\alpha=1}^{k_i} x_{i_\alpha}(t) L_{{\bm f}_i}\ltimes_{\alpha=1}^{k_i} x_{i_\alpha}(t)\\
&=M_{\oplus_i}L_{{\bm g}_i} \bar{\ltimes}_{\alpha=1}^{k_i} x_{i_\alpha}(t) L_{{\bm f}_i}(I_{2^{\iota_i-1}}\otimes M_{d,2})\bar{\ltimes}_{\alpha=1}^{k_i} x_{i_\alpha}(t)\\
&=M_{\oplus_i}L_{{\bm g}_i}\left(I_{2^{k_i+1}}\otimes (L_{{\bm f}_i}(I_{2^{\iota_i-1}}\otimes  M_{d,2}))\right)\\
&~~~~~~~~~~~~~~~~~~~~~~~~~~~~~~~~~~~~M_{r,2^{k_i+1}}\bar{\ltimes}_{\alpha=1}^{k_i} x_{i_\alpha}(t),
\end{aligned}
\end{equation}
where $M_{\oplus_i}\in\mathscr{L}_{2\times 4}$ and $L_{{\bm g}_i}\in\mathscr{L}_{2\times 2^{k_i+1}}$ are respectively the unknown structure matrices corresponding to logical coupling $\oplus_i$ and logical function ${\bm g}_i$ for each $i\in\mathrm{P}_1$.

For vertex $i\in\mathrm{P}_2$, the state feedback control ${\bm u}_i$ only depends on the states of local in-neighbors as
$${\bm u}_i(t)={\bm g}_i({\bm x}_{i_1}(t),{\bm x}_{i_2}(t),\cdots,{\bm x}_{i_{k_i}}(t)).$$
Consequently, we have that
\begin{equation*}
\begin{aligned}
x_i(t+1)&=M_{\oplus_i}u_i(t) f_i(x_{i_1}(t),x_{i_2}(t),\cdots,x_{i_{k_i}}(t))\\
&=M_{\oplus_i}L_{{\bm g}_i}(\ltimes_{\alpha=1}^{k_i} x_{i_\alpha}(t)) L_{{\bm f}_i}(\ltimes_{\alpha=1}^{k_i} x_{i_\alpha}(t))\\
&=M_{\oplus_i}L_{{\bm g}_i}\left(I_{2^{k_i}}\otimes L_{{\bm f}_i}\right)M_{r,2^{k_i}} (\ltimes_{\alpha=1}^{k_i} x_{i_\alpha}(t)).
\end{aligned}
\end{equation*}

Besides, for $i\in\mathrm{P}_3$, one can directly assign $\oplus_i=\wedge$ and ${\bm g}_i(t)={\bm x}_{\varpi_i}(t)$ if ${\bm f}_i=1$; otherwise, let $\oplus_i=\vee$ and ${\bm g}_i(t)={\bm x}_{\varpi_i}(t)$.

Finally, we respectively calculate the logical matrices $M_{\oplus_i}\in\mathscr{L}_{2\times 4}$, $L_{{\bm g}_i}\in\mathscr{L}_{2\times2^{k_i}}$, $i\in\mathrm{P}_1$ and $M_{\oplus_i}\in\mathscr{L}_{2\times 4}$, $L_{{\bm g}_i}\in\mathscr{L}_{2\times2^{k_i+1}}$, $i\in\mathrm{P}_2$ from the following equations
\begin{equation}\label{equation-equations}
\left\{\begin{aligned}
&M_{\oplus_i}L_{{\bm g}_i}\left(I_{2^{k_i+1}}\otimes (L_{{\bm f}_i}(I_{2^{\iota_i-1}}\otimes M_{d,2}))\right)M_{r,2^{k_i+1}}\\
&~~~~~~~~~~~~~~~~~~~~~~~~~=(\hat{A}_i\otimes {\bf 1}^\top_{2^{k_i}}) \mathrm{W}^\top_{\left[2,2^{\iota_i-1}\right]},~i\in\mathrm{P}_1,\\
&M_{\oplus_i}L_{{\bm g}_i}\left(I_{2^{k_i}}\otimes L_{{\bm f}_i}\right)M_{r,2^{k_i}}\\
&~~~~~~~~~~~~~~~~~~~~~~~~~=(\hat{A}_i\otimes {\bf 1}^\top_{2^{k_i-1}}) \mathrm{W}^\top_{\left[2,2^{\iota_i-1}\right]},~i\in\mathrm{P}_2,
\end{aligned}\right.
\end{equation}
with matrix $\hat{A}_i\in \{I_2,I'_2\}$ for all $i\in \mathrm{P}_1 \cup \mathrm{P}_2$.

It should be stressed that \cite{Liff2016TNNLS1585} has proved that the latter equation in (\ref{equation-equations}) is solvable. Besides, if we regard $L_{{\bm f}_i} (I_{2^{\iota_i-1}} \otimes M_{d,2})$ as a whole, the former equation of (\ref{equation-equations}) is still in the same form as the latter one. Thus, (\ref{equation-equations}) is also solvable. Once the structure matrices $M_{\oplus_i}$ and $L_{{\bm g}_i}$ are solved, we can obtain their logical form by the inverse transformation of Lemma \ref{lemma-structurematrix}, given in \cite{chengdz2011springer}.

\begin{theorem}
The pinning controlled BN (\ref{equation-pinningBN}) is observable.
\end{theorem}
\begin{proof}
According to equation (\ref{equation-equations}), we know that, the unique in-neighbor of vertex $v_i$ in the network structure of the pinning controlled BN (\ref{equation-pinningBN}), $i\in\mathrm{P}_1\cup\mathrm{P}_2$ is vertex $v_{\varpi_i}$, which is the desired in-neighbor on the desired observed paths $\mathrm{O}_1,\mathrm{O}_2,\cdots,\mathrm{O}_p$. This claim is supported by the following derivation:
$$\begin{aligned}
&(\hat{A}_i\otimes {\bf 1}^\top_{2^{k_i}}) \mathrm{W}^\top_{\left[2,2^{\iota_i-1}\right]} \bar{\ltimes}_{\alpha=1}^{k_i} x_{i_\alpha}(t)=(\hat{A}_i\otimes {\bf 1}^\top_{2^{k_i}}) \mathrm{W}_{\left[2^{\iota_i-1},2\right]} \bar{\ltimes}_{\alpha=1}^{k_i} x_{i_\alpha}(t)\\
&=(\hat{A}_i\otimes {\bf 1}^\top_{2^{k_i}}) x_{\varpi_i}(t) \ltimes_{\alpha=1}^{k_i} x_{i_\alpha}(t)=\hat{A}_ix_{\varpi_i}(t)
\end{aligned}$$
and
$$\begin{aligned}&(\hat{A}_i\otimes {\bf 1}^\top_{2^{k_i-1}}) \mathrm{W}^\top_{\left[2,2^{\iota_i-1}\right]} \ltimes_{\alpha=1}^{k_i} x_{i_\alpha}(t)\\
&=(\hat{A}_i\otimes {\bf 1}^\top_{2^{k_i-1}}) x_{\varpi_i}(t) (\ltimes_{\alpha=1}^{\varpi_i-1} x_{i_\alpha}(t)) (\ltimes_{\alpha=\varpi_i+1}^{k_i} x_{i_\alpha}(t))=\hat{A}_ix_{\varpi_i}(t)\end{aligned}$$

By Lemma \ref{lemma-basictheorem}, the pinning controlled BN (\ref{equation-pinningBN}) has both Properties $P_1$ and $P_2$. One, therefore, can conclude that the pinning controlled BN (\ref{equation-pinningBN}) is observable.
\end{proof}

\section{Construction of Sensors}\label{section-sensors}
The above pinning control approach indeed contributes to the sensors design of large-scale BNs, thus it eliminates the illogical assumption that the states of all system nodes should be accessible. In this section, we shall utilize the configurations $\oplus_i$, ${\bm g}_i$ and $\mathrm{P}$ designed in the above pinning controllers to construct the efficient sensors.

\subsection{Sensors Construction}
For every node $i\in\mathrm{P}$, we assume that its corresponding vertex locates on the observed path $\mathrm{O}_{\pi_i}$ and needs $\lambda_{i}$ steps to pass the information of $v_i$ to the directly observable vertex $v_{\pi_i}$ in network structure $\mathrm{G}$. Consequently, one can establish the following sensors:
\begin{equation}\label{equation-finalobserver}
\left\{\begin{aligned}
&\tilde{{\bm y}}_i(t) = {\bm g}_i({\bm x}_{i_1}(t),{\bm x}_{i_2}(t),\cdots,{\bm x}_{i_{k_i}}(t), {\bm x}_{\varpi_i}(t)),~i\in\mathrm{P}_1,\\
&\tilde{{\bm y}}_i(t) = {\bm g}_i({\bm x}_{i_1}(t),{\bm x}_{i_2}(t),\cdots,{\bm x}_{i_{k_i}}(t)),~i\in\mathrm{P}_2,\\
&\tilde{{\bm y}}_i(t) = {\bm x}_{\varpi_i}(t),~i\in\mathrm{P}_3.
\end{aligned}\right.
\end{equation}

On the basis of original outputs ${\bm y}_j$, $j\in[1,p]$, and sensors (\ref{equation-finalobserver}), the initial system state can be deduced from the output observations. We suppose that the desired observed path from vertex $v_i\in\mathrm{O}_{\pi_i}\backslash\{v_j|j\in\mathrm{P}_3\}$ to $v_{\pi_i}$ goes through $o_i$ vertices in the set $\mathrm{P}_1\cup\mathrm{P}_2$. Sequence $v_iv_{\xi^{\pi_i}_{\lambda_i}} v_{\xi^{\pi_i}_{\lambda_i-1}} \cdots v_{\xi^{\pi_i}_{1}}(:=v_{\pi_i})$ is the part of desired observed path from vertex $v_i$ to directly observable vertex $v_{\pi_i}$.

For the case of $o_i=0$, we can compute the initial state of node $i$ as ${\bm x}_i(0)={\bm f}_{\xi^{\pi_i}_{\lambda_i}} \circ {\bm f}_{\xi^{\pi_i}_{\lambda_i-1}} \circ \cdots\circ {\bm f}_{\xi^{\pi_i}_{1}}({\bm x}_{\pi_i}(\lambda_i))$\footnote{For two functions ${\bm f}^1:\mathscr{B}\rightarrow\mathscr{B}$ and ${\bm f}^2:\mathscr{B}\rightarrow\mathscr{B}$, composite function ${\bm f}^1\circ{\bm f}^2$ is defined as ${\bm f}^1\circ{\bm f}^2(\eta)={\bm f}^1({\bm f}^2(\eta))$ for any $\eta\in\mathscr{B}$.};

Next, we consider the scene of $o_i=1$, and a general situation with $o_i>1$ can be regarded as an extension. Without loss of generality, we assume that the desired observed path from $v_i$ to $v_{\pi_i}$ goes through the vertex $v_{\epsilon_i}$, $\epsilon_i\in\mathrm{P}_1\cup\mathrm{P}_2$, which is the out-neighbor of $v_i$ on the printed path. Thus, the initial state of node $i$ can be calculated as
$${\bm x}_i(0)=({\bm g}_{\epsilon_i}\oplus_{\epsilon_i}{\bm f}_{\epsilon_i} )\circ {\bm f}_{\xi^{\pi_i}_{\lambda_i-1}} \circ \cdots\circ {\bm f}_{\xi^{\pi_i}_{1}}({\bm x}_{\pi_i}(\lambda_i)).$$
Since, the state of every node can always be passed to the node corresponding to certain directly observable vertex like a shift register, we therefore can determine the initial state of any node. It enables us to reconstruct the entire state trajectory of BN (\ref{equation-BN}). Consequently, we can conclude the following theorem.
\begin{theorem}\label{theorem-observerdesign}
Via sensors (\ref{equation-finalobserver}), BN (\ref{equation-BN}) is observable.
\end{theorem}

\begin{remark}
For BN (\ref{equation-BN}) with $p=0$, we begin the procedure by appointing some nodes to be directly observed by sensors, and regard this part of output observations as original outputs.
\end{remark}

\subsection{Comparison with Traditional Approaches}\label{section-discussion}
One striking merit of this pinning control strategy is the reduction of computational complexity. Generally for the analysis and control of observability in BNs (see, e.g., \cite{chengdz2018scl22,chengdz2016scl76,guoyq2018tnnls6402}), the time complexity is $O(2^{2n})$ at least, which is computationally heavy to deal with the large-scale BNs. By our approach, the time complexity to determine the pinned nodes (cf. Algorithm \ref{algorithm-selectnode}) is linear in node number $n$ as $O(n^2)$; and for the part of designing feedback controllers and logical couplings we only address the $2\times 2^{d+1}$-dimensional logical matrix at most, where $d$ is the largest in-degree of vertices corresponding to pinned nodes. Thus, the total time complexity of this method is bounded by $O(n^2+n2^{d})$, which is lower than $O(2^{2n})$ since $d \ll n$. Moreover, some results (cf. \cite{jeong2000nature,jeong2001nature}) reveal that the large-scale biological networks are always sparsely connected, i.e., $d \ll n$.

In addition, by this method, it takes $O(n^2)$ time to identify the pinned nodes by the network structure $\mathrm{G}$ as mentioned above while the approaches in \cite{Liff2016TNNLS1585,liff2018tnnls} are bounded by $O(2^{2n})$. Additionally, the traditional pinning methods may inject the control inputs on all nodes. This case can be avoided by our method and even empirically pins the less nodes in most of scenes.

Compared with \cite{Margaliot2019TAC2727} and \cite{Margaliot2019IEEECSL210} that design the optimal and suboptimal sensors with respect to the number, we provide an alternative pinning observability approach to design the efficient sensors. In comparison with the observers design in \cite{Valcher2012TAC1390,zhangzh2020tac,yangjq2020tcns}, instead of the $2^n\times2^n$-dimensional network transition matrix $L$, the design of sensors here only adopts the $2\times 2^{d+1}$-dimensional logical matrix.

\section{Two Simulation Examples}\label{section-example}
\subsection{D. Melanogaster Segmentation Polarity Gene Network}
In the first subsection, we shall briefly study the D. melanogaster segmentation polarity gene network, which was modeled in \cite{xiao2007impact} and studied in \cite{Liff2016TNNLS1585}; see \cite{xiao2007impact} for a survey.

In D. melanogaster segmentation polarity gene network, six genes respectively represent the mRNA or protein: $\text{wg1}:{\bm x}_1$, $\text{wg2}:{\bm x}_2$ $\text{wg3}:{\bm x}_3$, $\text{wg4}:{\bm x}_4$, $\text{PTC1}:{\bm x}_5$ and $\text{PTC2}:{\bm x}_6$. The logical dynamics for the above six variables described in the following, and three outputs have also been given below
\begin{equation}\label{equation-exa1-nodedynamics}
\left\{\begin{aligned}
&{\bm x}_1(t+1)={\bm x}_1(t) \wedge \overline{{\bm x}_2(t)} \wedge \overline{{\bm x}_4(t)},\\
&{\bm x}_2(t+1)=\overline{{\bm x}_1(t)} \wedge \overline{{\bm x}_2(t)} \wedge \overline{{\bm x}_3(t)},\\
&{\bm x}_3(t+1)={\bm x}_1(t) \vee {\bm x}_3(t),~{\bm x}_4(t+1)={\bm x}_2(t) \vee {\bm x}_4(t),\\
&{\bm x}_5(t+1)=(\overline{{\bm x}_2(t)} \wedge \overline{{\bm x}_4(t)}) \vee ({\bm x}_5(t) \wedge \overline{{\bm x}_1(t)} \wedge \overline{{\bm x}_3(t)}),\\
&{\bm x}_6(t+1)=(\overline{{\bm x}_1(t)} \wedge \overline{{\bm x}_3(t)}) \vee ({\bm x}_6(t) \wedge \overline{{\bm x}_2(t)} \wedge \overline{{\bm x}_4(t)}),\\
&{\bm y}_1(t)={\bm x}_4(t),~{\bm y}_2(t)={\bm x}_5(t),~~{\bm y}_3(t)={\bm x}_6(t).
\end{aligned}\right.
\end{equation}

According to (\ref{equation-exa1-nodedynamics}), one can establish its network structure $\mathrm{G}$ as in Fig. \ref{fig-exa1-wringgraph}. While vertices $v_4$, $v_5$ and $v_6$ are directly observable vertices, there is a self loop on each vertex (see Fig. \ref{fig-exa1-wringgraph}). Thus, BN (\ref{equation-exa1-nodedynamics}) does not satisfy Properties $P_1$ and $P_2$, and Lemma \ref{lemma-basictheorem} cannot be used to determine its observability.
\begin{figure}[H]
\centering
\includegraphics[width=0.3\textwidth=0.5]{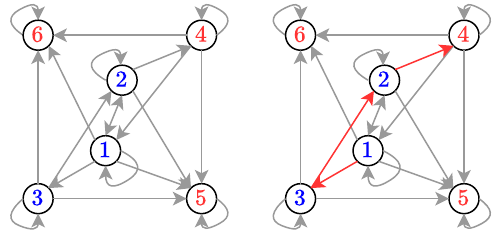}
\caption{Network structure of D. melanogaster segmentation polarity gene network (\ref{equation-exa1-nodedynamics}) is presented, where vertices $v_4$, $v_5$ and $v_6$ are directly observable vertices. The left subgraph is the original network structure, while the three red paths in the right subgraph are the desired observed paths provided by Algorithm \ref{algorithm-selectnode}.}\label{fig-exa1-wringgraph}
\end{figure}

Hence, we draw the state transition graph of this network in Fig. \ref{fig-6nodes-STG} and use [\cite{Margaliot2013Observability2351}, Th. 5] to check its observability, that is, whether its state transition graph has two separate cycles assigned with the same output sequence or contains a vertex having two different in-neighbors with the same output. Obviously, we can find the self loops $(101010)$ and $(001010)$, both of which correspond to outputs ${\bm y}_1=0$, ${\bm y}_2=1$ and ${\bm y}_3=0$ in the state transition graph. Thus, BN (\ref{equation-exa1-nodedynamics}) is not observable, and we devote to designing the pinning controller to make it be observable.
\begin{figure}[!ht]
\centering
\includegraphics[width=0.35\textwidth=0.5]{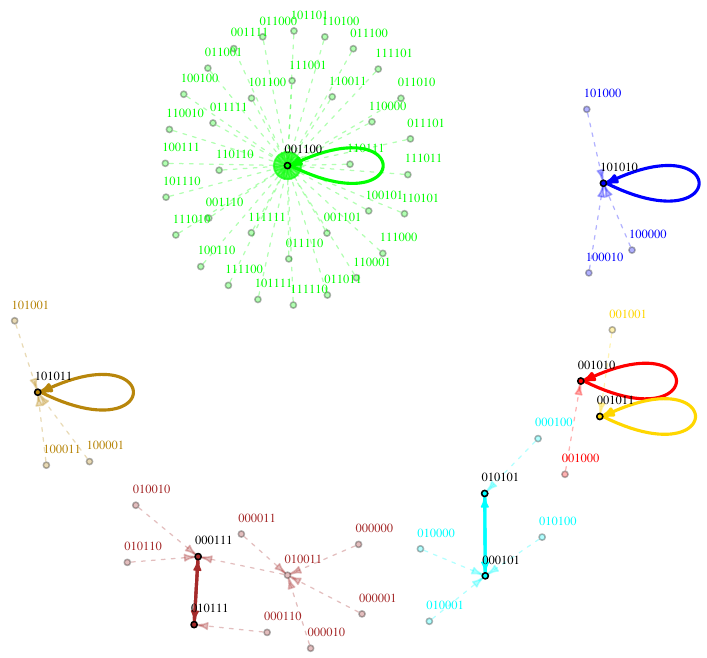}
\caption{State transition graph of D. melanogaster segmentation polarity gene network (\ref{equation-exa1-nodedynamics}) is drawn with seven different attractors.}\label{fig-6nodes-STG}
\end{figure}

To this end, we call Algorithm \ref{algorithm-selectnode} on the digraph $\mathrm{G}$, and three desired observed paths are generated as $\mathrm{O}_1: v_1v_3v_2v_4$, $\mathrm{O}_2: v_5$ and $\mathrm{O}_3: v_6$, which are described in the right subgraph of Fig. \ref{fig-exa1-wringgraph}. By Definition \ref{definition-observedpath}, one can select the pinned node set $\mathrm{P}=\mathrm{P}_2=\{2,3,4\}$.

Consider node $2\in\mathrm{P}_2$ and the second equation of (\ref{equation-equations}). Since its nodal dynamics is ${\bm f}_{2}(t)=\overline{{\bm x}_1(t)} \wedge \overline{{\bm x}_2(t)} \wedge \overline{{\bm x}_3(t)}$, one can calculate that its structure matrix is $L_{{\bm f}_2}=\delta_2[2,2,2,2,2,2,2,1]$. Again, vertex $v_2$ belongs to the observed path $\mathrm{O}_1$ and its desired in-neighbor is $v_3$, thus $\iota_2=3$ and $k_2=3$. Letting $\hat{A}=I_2$, one can solve that $M_{\oplus_{2}}=\delta_2[1,1,1,2]$ and $L_{{\bm g}_2}=\delta_2[1,2,2,1,1,2,2,1]$, which respectively correspond to $\oplus_2=\vee$ and ${\bm g}_2(t)=\overline{{\bm x}_3(t)}$. In a similar manner, we can obtain that $M_{\oplus_{3}}=\delta_2[2,1,1,2]$ and $L_{{\bm g}_3}=\delta_2[2,2,1,2]$, $M_{\oplus_{4}}=\delta_2[2,1,1,2]$ and $L_{{\bm g}_4}=\delta_2[2,2,1,2]$. Accordingly, one can obtain their logical form as $\oplus_{3}=\oplus_{4}=\bar{\vee}$, ${\bm g}_3(t)=\overline{{\bm x}_3(t) \rightarrow {\bm x}_1(t)}$, and ${\bm g}_4(t)=\overline{{\bm x}_4(t) \rightarrow {\bm x}_2(t)}$.

Therefore, by Theorem \ref{theorem-observerdesign}, sensors can be designed as follows:
\begin{equation}
\left\{\begin{aligned}
&{\bm y}_1(t)={\bm x}_4(t),~{\bm y}_2(t)={\bm x}_5(t),~~{\bm y}_3(t)={\bm x}_6(t),\\
&\tilde{{\bm y}}_1(t)=\overline{{\bm x}_3(t)},~\tilde{{\bm y}}_2(t)=\overline{{\bm x}_3(t) \rightarrow {\bm x}_1(t)},\\
&\tilde{{\bm y}}_3(t)=\overline{{\bm x}_4(t) \rightarrow {\bm x}_2(t)}.
\end{aligned}\right.
\end{equation}

\subsection{T-cell Receptor Kinetics Network}
In the second part, we consider a reduced BN called T-cell receptor kinetics with $37$ state nodes and $3$ input nodes, which was originally modeled in \cite{klamt2006methodology}. The detailed logical reactions are presented in (\ref{equation-BN-exa2}). Again, according to its nodal dynamics equation (\ref{equation-BN-exa2}), we can draw its network structure $\mathrm{G}$ as in Fig. \ref{fig-exa2-wringdigraph}. Even more, the original sensors are ${\bm y}_1(t)={\bm x}_1(t)$, ${\bm y}_2(t)={\bm x}_{22}(t)$, and ${\bm y}_3(t)={\bm x}_{25}(t)$, which respectively observe the values of nodes AP1, LAT and NFAT.

\begin{figure*}[!ht]
\centering
\begin{equation}\label{equation-BN-exa2}
\begin{array}{ll}
\text{CD8}: {\bm u}_{1}, \text{CD45}: {\bm u}_2, \text{TCRlig}: {\bm u}_3,&\text{AP1}: {\bm x}_1(\ast)={\bm x}_9\wedge{\bm x}_{18}\\
\text{Ca}: {\bm x}_2(\ast)={\bm x}_{14},&\text{Calcin}: {\bm x}_3(\ast)={\bm x}_2, \\
\text{cCbl}: {\bm x}_4(\ast)={\bm x}_{37},&\text{CRE}: {\bm x}_5(\ast)={\bm x}_6, \\
\text{CREB}: {\bm x}_6(\ast)={\bm x}_{32},&\text{DAG}: {\bm x}_7(\ast)={\bm x}_{25}, \\
\text{ERK}: {\bm x}_8(\ast)={\bm x}_{21},&\text{Fos}: {\bm x}_9(\ast)={\bm x}_8, \\
\text{Fyn}: {\bm x}_{10}(\ast)=({\bm x}_{20}\wedge{\bm u}_2(t))\vee({\bm x}_{35}(t)\wedge{\bm u}_2),&\text{Gads}: {\bm x}_{11}(\ast)={\bm x}_{19},\\
\text{Grb2Sos}: {\bm x}_{12}(\ast)={\bm x}_{19}, & \text{IKKbeta}: {\bm x}_{13}(\ast)={\bm x}_{24},\\
\text{IP3}: {\bm x}_{14}(\ast)={\bm x}_{25}, & \text{Itk}: {\bm x}_{15}(\ast)={\bm x}_{34}\wedge{\bm x}_{37},\\
\text{IkB}: {\bm x}_{16}(\ast)=\overline{{\bm x}_{13}}, & \text{JNK}: {\bm x}_{17}(\ast)={\bm x}_{33},\\
\text{Jun}: {\bm x}_{18}(\ast)={\bm x}_{17}, & \text{LAT}: {\bm x}_{19}(\ast)={\bm x}_{37},\\
\text{Lck}: {\bm x}_{20}(\ast)=\overline{{\bm x}_{26}}\wedge {\bm u}_1 \wedge {\bm u}_2, &\text{MEK}: {\bm x}_{21}(\ast)={\bm x}_{28},\\
\text{NFAT}: {\bm x}_{22}(\ast)={\bm x}_3, & \text{NFkB}: {\bm x}_{23}(\ast)=\overline{{\bm x}_{16}},\\
\text{PKCth}: {\bm x}_{24}(\ast)={\bm x}_7, & \text{PLCg(act)}: {\bm x}_{25}(\ast)=({\bm x}_{15}\wedge {\bm x}_{27}\wedge {\bm x}_{34} \wedge {\bm x}_{37})\vee({\bm x}_{27}\wedge{\bm x}_{31}\wedge{\bm x}_{34}\wedge{\bm x}_{37})\\
\text{PAGCsk}: {\bm x}_{26}(\ast)={\bm x}_{10}\vee \overline{{\bm x}_{35}}, &\text{PLCg(bind)}: {\bm x}_{27}(\ast)={\bm x}_{19},\\
\text{Raf}: {\bm x}_{28}(\ast)={\bm x}_{29}, & \text{Ras}: {\bm x}_{29}(\ast)={\bm x}_{12}\vee{\bm x}_{14},\\
\text{RasGRP1}: {\bm x}_{30}(\ast)={\bm x}_7\wedge{\bm x}_{13}, & \text{Rlk}: {\bm x}_{31}(\ast)={\bm x}_{20},\\
\text{Rsk}: {\bm x}_{32}(\ast)={\bm x}_8, & \text{SEK}: {\bm x}_{33}(\ast)={\bm x}_{24},\\
\text{SLP76}: {\bm x}_{34}(\ast)={\bm x}_{11}, & \text{TCRbind}: {\bm x}_{35}(\ast)=\overline{{\bm x}_4}\wedge{\bm u}_3,\\
\text{TCRphos}: {\bm x}_{36}(\ast)={\bm x}_{10}\vee({\bm x}_{20}\wedge {\bm x}_{35}) & \text{ZAP70}: {\bm x}_{37}(\ast)=\overline{{\bm x}_4}\wedge{\bm x}_{20}\wedge{\bm x}_{36}.
\end{array}
\end{equation}
\hrulefill
\vspace*{1pt}
\end{figure*}

\begin{figure}[!ht]
\centering
\includegraphics[width=0.48\textwidth=0.5]{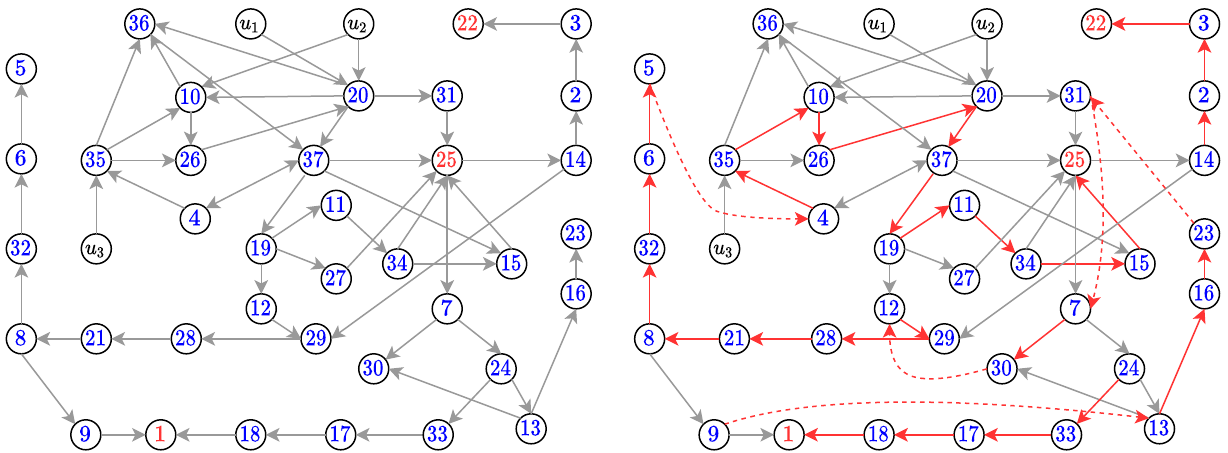}
\caption{Network structure $\mathrm{G}$ of T-cell receptor kinetics is presented, where the red vertices stand for the directly observable vertices. The left subgraph is the original network structure, and the right one is provided by Algorithm \ref{algorithm-selectnode}. More precisely, three red paths in the right subgraph are the desired observed paths, and the dashed edges therein do not exist in the original network structure $\mathrm{G}$.}\label{fig-exa2-wringdigraph}
\end{figure}

Since this BN does not satisfy both Properties $P_1$ and $P_2$, we cannot judge the observability of this BN from the viewpoint of network structure. Hence, we draw its state transition graph in Fig. \ref{fig-exa2-STG} by appointing $(1000110101110001000000000110000000010)$ as the initial state. Observing that the out-neighbors of vertices $A$ and $B$ in state transition graph are both $(0000000101000001000000000100001000111)$, BN (\ref{equation-BN-exa2}) is not observable according to [\cite{Margaliot2013Observability2351}, Th. 5].
\begin{figure}[htb]
\centering
\includegraphics[width=0.35\textwidth=0.5]{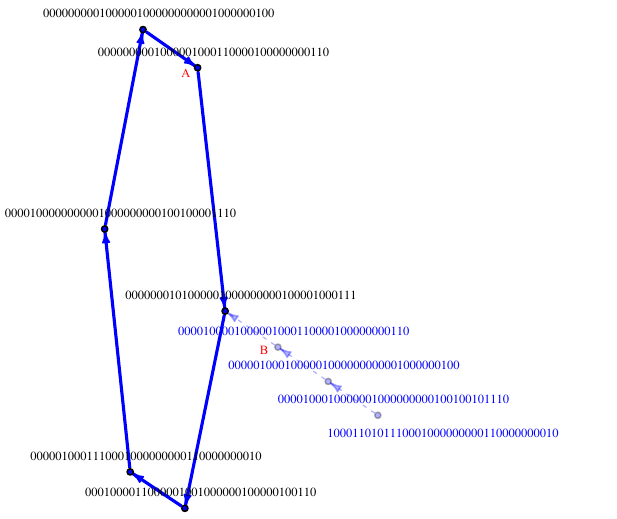}
\caption{The state transition graph of T-cell receptor kinetics from initial state $(1000110101110001000000000110000000010)^\top$ is presented.}\label{fig-exa2-STG}
\end{figure}

In what follows, we proceed to design the sensors under which BN (16) is observable. Algorithm \ref{algorithm-selectnode} yields three desired observed paths as $\mathrm{O}_1:v_{24}v_{33}v_{17}v_{18}v_{1}$, $\mathrm{O}_2:v_{14}v_{2}v_{3}v_{22}$ and $\mathrm{O}_3: v_9v_{13}v_{16}v_{23}v_{31}v_{7}v_{30}v_{12}v_{29}v_{28}v_{21}v_8v_{32}v_6v_5$ $v_4v_{35}v_{10}v_{26}v_{20}v_{37}v_{19}v_{11}v_{34}v_{15}v_{25}$, which are colored by red in the right subgraph of Fig. \ref{fig-exa2-wringdigraph}. From Definition \ref{definition-observedpath}, its pinned node set $\mathrm{P}$ is selected as $\mathrm{P}:=\{1,25,15,37,20,26,10,35,4,29,12,30,7,31,13\}$, which occupies approximately $40.54\%$ of all network nodes. Moreover, pinned node set $\mathrm{P}$ can be classified into two distinct types: $\mathrm{P}_1=\{12,7,31,13,4\}$ and $\mathrm{P}_2=\{1,25,15,37,20,26,10,35,29,30\}$.

Considering system node $7\in\mathrm{P}_1$, due to ${\bm x}_7(t+1)={\bm x}_{25}(t)$, one can claim that its structure matrix is $L_{{\bm f}_7}=I_2$. Again, vertex $v_7$ lies in the observed path $\mathrm{O}_3$ and its desired in-neighbor on the observed path $\mathrm{O}_3$ is $v_{31}$, thus one can utilize the first equation of (\ref{equation-equations}) to design the state feedback gain ${\bm g}_{7}$. Letting $\iota_{7}=2$ and $\hat{A}=I_2$, one can solve that $M_{\oplus_{7}}=\delta_2[2,1,1,2]$ and $L_{{\bm g}_7}=\delta_2[2,1,1,2]$, which correspond to dynamics $\oplus_7=\bar{\vee}$ and ${\bm g}_7=\overline{{\bm x}_{25}}\leftrightarrow {\bm x}_{31}$. Similarly, one can construct logical couplings $\oplus_i$ and state feedback gains ${\bm g}_i$ for other nodes in $\mathrm{P}_1$ as $\oplus_{12}=\bar{\vee}$, ${\bm g}_{12}={\bm x}_{30} \leftrightarrow \overline{{\bm x}_{19}}$, $\oplus_{13}=\bar{\vee}$, ${\bm g}_{13}={\bm x}_{9} \leftrightarrow \overline{{\bm x}_{24}}$, $\oplus_{31}=\bar{\vee}$, ${\bm g}_{31}={\bm x}_{23} \leftrightarrow \overline{{\bm x}_{20}}$, and $\oplus_4=\bar{\vee}$, ${\bm g}_4=\overline{{\bm x}_{37}}\leftrightarrow {\bm x}_{5}$.

Then, we analyze the system nodes in the set $\mathrm{P}_2$. For instance, considering vertex $v_{1}$, it locates on the desired observed path $\mathrm{O}_1$ and its desired in-neighbor is vertex $v_{18}$. Therefore, according to the second equation in (\ref{equation-equations}), letting $k_1=2$, $\iota_7=2$ and $\hat{A}_7=I_2$, one can calculate that $M_{\oplus_1}=\vee$ and $L_{{\bm g}_1}=\delta_2[1,2,1,2]$, which corresponds to $\oplus_1=\vee$ and ${\bm g}_{1}={\bm x}_{18}$. Moreover, for other nodes in $\mathrm{P}_2$, we can similarly compute as $\oplus_{15}=\wedge$, ${\bm g}_{15}={\bm x}_{25}$, $\oplus_{16}=\bar{\vee}$, ${\bm g}_{16}={\bm x}_{35}\leftrightarrow \overline{({\bm x}_{20}\vee{\bm u}_2)\vee({\bm x}_{35}\wedge{\bm u}_2)}$, $\oplus_{20}=\wedge$, ${\bm g}_{20}=\overline{{\bm x}_{26}}$, $\oplus_{25}=\bar{\vee}$, ${\bm g}_{25}={\bm x}_{15} \leftrightarrow \overline{({\bm x}_{15}\wedge{\bm x}_{27}\wedge{\bm x}_{34}\wedge{\bm x}_{37})\vee({\bm x}_{27}\wedge{\bm x}_{31}\wedge{\bm x}_{34}\wedge{\bm x}_{37})}$, $\oplus_{26}=\bar{\vee}$, ${\bm g}_{26}=\overline{\overline{{\bm x}_{35}}\rightarrow{\bm x}_{10}}$, $\oplus_{29}=\bar{\vee}$, ${\bm g}_{29}=\overline{{\bm x}_{14}\rightarrow{\bm x}_{12}}$, $\oplus_{30}=\wedge$, ${\bm g}_{30}={\bm x}_7$, and $\oplus_{37}=\wedge$, ${\bm g}_{37}={\bm x}_{20}$.

In BN (\ref{equation-BN-exa2}), about $40.54\%$ nodes are pinned. Additionally, during the whole process, our method only needs to address a $2\times2^5$-dimensional matrix instead of a $2^{37}\times 2^{37}$-dimensional one. By Theorem \ref{theorem-observerdesign}, sensors can be designed as
\begin{equation*}
\begin{aligned}
&{\bm y}_1(t)={\bm x}_1(t),~{\bm y}_2(t)={\bm x}_{22}(t),~{\bm y}_3(t)={\bm x}_{25}(t),\\
&\tilde{{\bm y}}_1=\overline{{\bm x}_{25}}\leftrightarrow {\bm x}_{31},~\tilde{{\bm y}}_2={\bm x}_{30} \leftrightarrow \overline{{\bm x}_{19}},~\tilde{{\bm y}}_3={\bm x}_{9} \leftrightarrow \overline{{\bm x}_{24}},\\
&\tilde{{\bm y}}_4={\bm x}_{23} \leftrightarrow \overline{{\bm x}_{20}},~\tilde{{\bm y}}_5=\overline{{\bm x}_{37}}\leftrightarrow {\bm x}_{5},~\tilde{{\bm y}}_6={\bm x}_{18},~\tilde{{\bm y}}_7={\bm x}_{25},\\
&\tilde{{\bm y}}_8={\bm x}_{35}\leftrightarrow \overline{({\bm x}_{20}\vee{\bm u}_2)\vee({\bm x}_{35}\wedge{\bm u}_2)},~\tilde{{\bm y}}_9=\overline{{\bm x}_{26}},\\
&\tilde{{\bm y}}_{10}={\bm x}_{15} \leftrightarrow \overline{({\bm x}_{15}\wedge{\bm x}_{27}\wedge{\bm x}_{34}\wedge{\bm x}_{37})\vee({\bm x}_{27}\wedge{\bm x}_{31}\wedge{\bm x}_{34}\wedge{\bm x}_{37})},\\
&\tilde{{\bm y}}_{11}=\overline{\overline{{\bm x}_{35}}\rightarrow{\bm x}_{10}},~\tilde{{\bm y}}_{12}=\overline{{\bm x}_{14}\rightarrow{\bm x}_{12}},~\tilde{{\bm y}}_{13}={\bm x}_7,~\tilde{{\bm y}}_{14}={\bm x}_{20}.
\end{aligned}
\end{equation*}

\section{Conclusion}\label{section-conclusion}
In this article, we have proposed a sensor design approach for large-scale BNs via pinning observability on the basis of the observability criteria in \cite{Margaliot2019TAC2727} and \cite{Margaliot2019IEEECSL210}. To this end, a novel pinning control strategy was developed for the observability of BNs. Based on the network structures, and without using ASSR framework, an algorithm was developed to produce a series of desired observed paths, and distinct types of pinned nodes were picked according to whether or not the corresponding vertices satisfy the definition of observed paths. Once pinned nodes have been picked, state feedback gain was designed for each of them. As a result, the extra sensors can be added as the state feedback control inputs in the pinning controller. The time complexity of our method was totally $O(n^2+n2^d)$, which is applicable for large-scale biological networks. Finally, the sensors design for D. melanogaster segmentation polarity gene network with $6$ state nodes and T-cell receptor kinetics with $37$ state nodes was successfully addressed.

\section*{Acknowledgment}
We are sincerely grateful to the experts in the graduate academic forum at the 2nd Workshop of TCCT Logical Control Systems for wealthy discussions and helpful suggestions.



\end{document}